\documentclass[12pt]{article}
\usepackage{fullpage,amssymb,graphicx,amsmath}

\usepackage{orcidlink}
\usepackage{hyperref,url}

\def\epi{\mathrm{epi}}

\def\dom{{\mathrm{dom}}}

\def\bbR{\mathbb{R}}
\def\calX{\mathcal{X}}
\def\calL{\mathcal{L}}

\def\calT{\mathcal{T}}
\def\inner#1#2{{\langle #1,#2 \rangle}}
\def\Inner#1#2{{\left\langle #1,#2 \right\rangle}}
\def\GL{\mathrm{GL}}

\def\st{\ :\ }
\def\eqdef{:=}

\def\st{{\ :\ }}

\def\barR{{\overline{\mathbb{R}}}}
\def\bbP{\mathbb{P}}

\def\calP{{\mathcal{P}}}

\newenvironment{proof}{\paragraph{Proof:}}{\hfill$\square$}

\newtheorem{Example}{Example}
\newtheorem{Theorem}{Theorem}
\newtheorem{Definition}{Definition}

\newtheorem{Property}{Property}
\newtheorem{Proposition}{Proposition}

\title{A note on the  Artstein-Avidan--Milman's generalized Legendre transforms}

\author{Frank Nielsen~\orcidlink{0000-0001-5728-0726}\\ \ \\ Sony Computer Science Laboratories Inc.\\ Tokyo, Japan}

\date{} 

\begin{document}
\maketitle

\begin{abstract}
Artstein-Avidan and Milman [Annals of mathematics (2009), (169):661-674] characterized invertible reverse-ordering transforms on the space of  
lower semi-continuous extended real-valued convex functions as 
affine deformations of the ordinary Legendre transform.
In this work, we first prove that all those generalized Legendre transforms on functions 
correspond to the ordinary Legendre transform on dually corresponding affine-deformed functions:
In short, generalized convex conjugates are ordinary convex conjugates of dually affine-deformed functions.
Second, we explain how these generalized Legendre transforms can be derived from the dual Hessian structures of information geometry.
\end{abstract}

\section{Introduction}

Let $\bbR$ denote the field of real numbers and $\barR=\bbR\cup\{\pm\infty\}$ the extended real line.
A $m$-variate extended function $F:\bbR^m\rightarrow\barR$ is proper when its efficient domain $\Theta=\dom(F):=\{\theta\st F(\theta)<+\infty\}$ is non-empty,
 and lower semi-continuous (lsc.) when its epigraph $\epi(F):=\{ (\theta,y) \st y\geq F(\theta),\ \theta\in \dom(F)\}$ is closed with respect to the metric topology of $\bbR^m$.
We denote by $\Gamma_0$ (shorcut for $\Gamma_0(\bbR^m)$) the space of of proper lower semi-continuous extended real-valued convex functions.
 
Consider the Legendre-Fenchel transform~\cite{bauschke2012fenchel} (LFT) $\calL F$ of a function $F$:
 
\begin{equation}
(\calL F)(\eta)\eqdef \sup_{\theta\in\bbR^m} \left\{\inner{\theta}{\eta}-F(\theta)   \right\},
\end{equation}
where $\inner{v}{v'}\eqdef\sum_{i=1}^m v_i\, v_i'$ denote the Euclidean inner product.
The function $F^*=\calL F$ is called the convex conjugate function. 
Some examples of convex conjugates are reported in Appendix~\S\ref{sec:examples}.

The Moreau-Fenchel-Rockafellar theorem~\cite{Correa-2023} states that when a function $F\in\Gamma_0$, its biconjugate function $({F^*})^*$ coincides with the function $F$: ${F^*}^*=F$. That is, the Legendre transform is an involutive transform on $\Gamma_0$.
For general lsc. functions $F$ (possibly non-convex), it is known that the biconjugate function $(F^*)^*\eqdef\calL(\calL F)$ lower bounds the function $F$ by the largest possible lsc. convex function: $(F^*)^*\leq F$. This lower bound convexification property has been proven useful in machine learning~\cite{SAM-2023}.

This study was motivated by the fundamental result of Artstein-Avidan and Milman~\cite{GenLegendreFenchel-2007,AxiomatizationLegendreTransform-2009} who proved the following theorem:
 
\begin{Theorem}[\cite{AxiomatizationLegendreTransform-2009}, Theorem~7]
Let $\calT$ be an invertible transform such that
\begin{itemize}
\item $F_1\leq F_2 \Rightarrow \calT F_2\leq \calT F_1$ and
\item $\calT F_1\leq \calT F_2 \Rightarrow F_2\leq F_1$.
\end{itemize}
Then $\calT$ is a generalized Legendre-Fenchel transform (GLFT), written canonically as:
\begin{equation}\label{eq:aam}
(\calT F)(\eta)=\lambda (\calL F)(E\eta+f)+\inner{\eta}{g}+h,
\end{equation}
where $\lambda>0$, $E\in\GL(\bbR^m)$ (general linear group), $f,g\in\bbR^m$ and $h\in\bbR$.
\end{Theorem}

We note in passing that Fenchel~\cite{fenchel1949conjugate,fenchel2013conjugate} interpreted the graph of the Legendre transform of a $m$-variate function $F\in\Gamma_0$  as the polarity with respect to the paraboloid surface of $\bbR^{m+1}$
$$
\calP= \left\{ \left(\theta,y=\frac{1}{2}\sum_{i=1}^m \theta_i^2\right) \ :\ \theta\in\bbR^m\right\}\in\bbR^{m+1},
$$  
of the graph of $F$. 
This geometric polarity connection of the Legendre transform is all the more interesting since 
B{\"o}r{\"o}czky and Schneider~\cite{boroczky2008characterization} characterized 
the duality of convex bodies in Euclidean spaces containing the origin in the interior into the same space~\cite{GenLegendreFenchel-2007}.

In this work, we shall first prove that a generalized convex conjugate $\calT F$ obtained by Eq.~\ref{eq:aam} can always be expressed as
the ordinary convex conjugate of a corresponding affine-deformed function. Namely, we shall show that
$$
(\calT F)(\eta)= \calL \left( \lambda F(A\theta+b)+\inner{\theta}{c}+d )\right),
$$
where $A\in\GL(\bbR^m)$, $b,c\in\bbR^m$ and $d\in\bbR$ are defined according to $E,f,g,h$ (details reported in Theorem~\ref{prop:GLFTasLFT}).
This equivalence result allows us to interpret the origin of the generalized Legendre transforms from the viewpoint of  information geometry~\cite{IG-2016,ay2017information}, 
and to untangle the various degrees of freedom used when defining
 generalized Legendre transforms in \S\ref{sec:ig}.

\section{Generalized Legendre transforms as ordinary Legendre transforms}

Consider a  parameter 
$$
P=(\lambda,A,b,c,d)\in \bbP \eqdef \bbR_{>0}\times \GL(\bbR^m)\times \bbR^m\times \bbR^m\times\bbR,
$$ 
and deform a function $F(\theta)$ by carrying affine transformations on both the parameter argument and its output as follows:

\begin{equation}\label{eq:Fparam}
F_{P}(\theta)\eqdef \lambda\, F(A\theta+b)+\inner{\theta}{c}+d.
\end{equation}
 
Those affine deformations preserve convexity:

\begin{Property}[Convexity-preserving affine deformations]\label{prop:conservecvx}
Let $F\in \Gamma_0$. Then  $F_P$ belongs to $\Gamma_0$ for all $P\in\bbP$. 
\end{Property}

\begin{proof}
Let us check the convexity of $F_P$:
$$
\alpha F_P(\theta_1) + (1-\alpha) F_P(\theta_2) = 
\lambda (\alpha F(\bar\theta_1) + (1-\alpha) F(\bar\theta_2))+\inner{\alpha\theta_1+(1-\alpha)\theta_2)}{c} + d, 
$$
where $\bar\theta_1\eqdef A\theta_1+b$ and $\bar\theta_2\eqdef A\theta_2+b$.
Since $F$ is convex, we have $\alpha F(\bar\theta_1)+ (1-\alpha) F(\bar\theta_2)\leq F(\alpha \bar\theta_1+ (1-\alpha)\bar\theta_2)$.
Let $\bar\theta_\alpha' \eqdef \alpha \bar\theta_1+ (1-\alpha)\bar\theta_2=A\bar\theta_\alpha'+b$ with $\theta_{\alpha}' \eqdef \alpha\theta_1+(1-\alpha)\theta_2$.
We get
\begin{eqnarray*}
\alpha F_P(\theta_1) + (1-\alpha) F_P(\theta_2)   &\leq& \underbrace{\lambda F(A\bar\theta_\alpha'+b) +\inner{\theta_{\alpha}'}{c} +d}_{=F_P(\alpha\theta_1+(1-\alpha)\theta_2)}, 
\end{eqnarray*}
hence proving that $F_P$ is convex. 
Since the lower semi-continuous property is ensured (i.e., epigraph of $F_P$ is closed), we conclude that $F_P\in\Gamma_0$.
\end{proof}

Next, let us express the convex conjugate of  a function $F_P$ according to the ordinary convex conjugate $F^*=\calL F$:

\begin{Proposition}[LFT of an affine-deformed function]\label{prop:LFTaffine}
The Legendre transform of $F_P$ with $P=(\lambda,A,b,c,d)\in\bbP$ is 
$$
(F_P)^*=(F^*)_{P^\diamond}
$$ 
where
\begin{equation}
P^\diamond\eqdef \left(
\lambda,\frac{1}{\lambda }A^{-1},
-\frac{1}{\lambda}A^{-1}c,
-A^{-1}b,
\inner{b}{A^{-1}c}-d
\right)\in\bbP.
\end{equation}
That is, we have $\calL (F_P)=(\calL F)_{P^\diamond}$.
\end{Proposition}

\begin{proof}
Let $\Gamma_1\subset\Gamma_0$ be defined as the class of  differentiable strictly convex  Legendre-type functions~\cite{LegendreType-1967} (see Appendix~\ref{sec:ltf}).
In general, given a function $F(\theta)\in\Gamma_1$, we proceed as follows to calculate its convex conjugate:
First, we find the inverse function $(\nabla F)^{-1}$  of its gradient $\nabla F$ to obtain the gradient of the convex conjugate: $\nabla F^*=(\nabla F)^{-1}$.
Then we have 
$$
F^*(\eta)=\inner{\theta}{\eta}-F(\theta)=\inner{(\nabla F)^{-1}(\eta)}{\eta}-F((\nabla F)^{-1}(\eta)).
$$

For parameter $P=(\lambda,A,b,c,d)$, consider the strictly convex and differentiable function 
$F_P(\theta) = \lambda F(A\theta+b)+ \inner{c}{\theta}+d$, 
for invertible matrix $A\in\GL(d,\bbR)$, vectors  $b, c\in\bbR^d$ and scalars $d\in\bbR$ and $\lambda\in\bbR_{>0}$.
The gradient of $F_P$ is
$$
\eta=\nabla F_P(\theta) = \lambda A^\top  {\nabla F}(A\theta+b)+c.
$$

Denote by $G=F^*$ and $G_P=(F_P)^*$ be the Legendre convex conjugates of $F$ and $F_P$, respectively. 
By solving the equation $\nabla F_P(\theta)=\eta$, we get the reciprocal gradient $\theta(\eta)=\nabla G_P(\eta)$: 
$$
\nabla G_P(\eta) = A^{-1}\,\nabla G\left( \frac{1}{\lambda} A^{-\top} (\eta-c) \right)-b.
$$

Therefore the Legendre convex conjugate is obtained as  
\begin{eqnarray*}
G_P(\eta) &=& \Inner{\eta}{\nabla G_P(\eta)}-F_P(\nabla G_P(\eta)),\\
&=& \lambda'\,  G(A'\eta+b')+\inner{c'}{\eta}+d',
\end{eqnarray*}
where
\begin{eqnarray*}
\lambda' &=& \lambda,\\
A' &=& \frac{1}{\lambda }A^{-1},\\
b' &=& -\frac{1}{\lambda}A^{-1}c,\\
c' &=& -A^{-1}b,\\
d' &=& \inner{b}{A^{-1}c}-d.
\end{eqnarray*}
Hence, it follows that we have $P^\diamond=\left(
\lambda,\frac{1}{\lambda }A^{-1},
-\frac{1}{\lambda}A^{-1}c,
-A^{-1}b,
\inner{b}{A^{-1}c}-d
\right)\in\bbP.$
\end{proof}

Let us check that the (diamond) $\diamond$-operator on affine deformation parameters is an involution:

\begin{Proposition}[$\diamond$-involution]\label{prop:diamondinvolution}
The parameter transformation $P^\diamond$ is an involution: $({P^{\diamond}})^\diamond=P$.
\end{Proposition}

\begin{proof}
Let $P=(\lambda,A,b,c,d)$ and $P^\diamond=(\lambda',A',b',c',d')$ with 
\begin{eqnarray*}
\lambda'&=&\lambda,\\
A' &=& \frac{1}{\lambda }A^{-1},\\
b'&=&-\frac{1}{\lambda}A^{-1}c,\\
c'&=&-A^{-1}b,\\
d'&=& \inner{b}{A^{-1}c}-d.
\end{eqnarray*}

We check that
${P^\diamond}^\diamond=(\lambda'',A'',b'',c'',d'')=(\lambda,A,b,c,d)=P$ component-wise as follows:

\begin{alignat*}{4}
& \lambda'' &=& \lambda' &=& \lambda,\\
& A'' &=& \frac{1}{\lambda }{A'}^{-1} &=& \frac{1}{\lambda}\left(\frac{1}{\lambda}A^{-1}\right)^{-1} &=& A,\\
& b'' &=& -\frac{1}{\lambda}{A'}^{-1}c' &=& -\frac{1}{\lambda}\left(\frac{1}{\lambda}A^{-1}\right)^{-1}(-A^{-1}b) &=& b,\\
& c'' &=& -{A'}^{-1}b' &=& -\left(\frac{1}{\lambda}A^{-1}\right)^{-1} \left(-\frac{1}{\lambda}A^{-1}c\right) &=& c,\\
& d''&=& \inner{b'}{{A'}^{-1}c'}-d' &=& \Inner{-\frac{1}{\lambda}A^{-1}c}{\frac{1}{\lambda}\left(\frac{1}{\lambda}A^{-1}\right)^{-1}(-A^{-1}b)}-\inner{b}{A^{-1}c}+d &=& d.
\end{alignat*} 
\end{proof}

The $\diamond$-involution confirms that the Legendre-Fenchel transform is an involution: 
$$
((F_P)^*)^*=\calL (\calL F_P)=\calL F_{P^\diamond}=F_{({P^{\diamond}})^\diamond}=F_P.
$$

Let us now define the following notation to express the generalized Legendre convex conjugates:

\begin{Definition}[Generalized Legendre-Fenchel convex conjugates~\cite{AxiomatizationLegendreTransform-2009}]
Let $\calL_{\lambda, E,f,g,h}$ denote a generalized Legendre-Fenchel transform:
$$
\calL_{\lambda, E,f,g,h} F \eqdef \calL_P F\eqdef \lambda (\calL F)(E\eta+f)+\inner{\eta}{g}+h
$$
  
for the parameter $P=(\lambda, E,f,g,h)\in\bbP$.
\end{Definition}

Our result is that we can interpret those generalized Legendre-Fenchel transforms of Eq.~\ref{eq:aam} as the ordinary Legendre transform on corresponding affine-deformed convex functions: 

\begin{Theorem}\label{prop:GLFTasLFT}
For any $F\in\Gamma_0$, we have $\calL_P(F)\eqdef(F^*)_P=\calL \left( F_{P^\diamond} \right)$.
\end{Theorem}

\begin{proof}
By definition, $\calL_P F\eqdef(\calL F)_P$.
Since $P=({P^\diamond})^\diamond$ (Proposition~\ref{prop:diamondinvolution}), we 
have $(\calL F)_P = (\calL F)_{({P^\diamond})^\diamond}$, 
and by using Proposition~\ref{prop:LFTaffine},
we get $(\calL F)_{({P^\diamond})^\diamond}=\calL(F_{P^\diamond})$.
To summarize, we have:
$$
\calL_P(F) = (\calL F)_P =   (\calL F)_{({P^\diamond})^\diamond} = \calL \left( F_{P^\diamond} \right). 
$$
\end{proof}

That is, in plain words, the Artstein-Avidan-Milman's generalized Legendre transforms~\cite{GenLegendreFenchel-2007,AxiomatizationLegendreTransform-2009} are ordinary Legendre transforms on affine-deformed functions evaluated on affine deformed arguments.

Next, we describe an information-geometric interpretation of Theorem~\ref{prop:GLFTasLFT} which casts light on the origin and meanings of the degrees of freedom used to define generalized Legendre transforms.

\section{An information-geometric interpretation of generalized Legendre transforms}\label{sec:ig}

We can interpret the fact that generalized Legendre convex conjugates are affine-deformed convex conjugates from the lens of information geometry~\cite{IG-2016}.
Consider a smooth strictly convex function $F:\bbR^m\rightarrow\barR$ of the space  $\Gamma_2$ of twice-differentiable lsc. Legendre-type convex functions (with $\Gamma_2\subset\Gamma_1\subset\Gamma_0$).
This function $F(\theta)$ defines a dually flat $m$-dimensional space~\cite{IG-2016} $(M,g,\nabla,\nabla^*)$ which is a global chart manifold $M=\{p\}$ of points $p$ equipped with   dual Hessian structures~\cite{Shima-2007} $(g,\nabla)$ and $(g,\nabla^*)$.
The primal Hessian structure is induced by a potential function $\psi$ on $M$ with a torsion-free flat affine connection $\nabla$  such that $\psi(p)=F(\theta(p))$ where $(M,\theta(\cdot))$ is the primal $\nabla$-affine coordinate system. 
The dual Hessian structure is induced by a dual potential function $\phi$ on $M$ with a dual torsion-free flat affine connection $\nabla^*$  such that $\phi(p)=F^*(\eta(p))$ where $(M,\eta(\cdot))$ is the dual $\nabla^*$-affine coordinate system. 
The Riemannian metric $g$ can be expressed in the dual coordinate charts as $g(\theta)=\nabla^2 F(\theta)$ or $g(\eta)=\nabla^2 F^*(\eta)$ where $F^*$ is the convex conjugate.
These two potential functions $\psi$ and $\phi$ living on the manifold $M$ satisfy the   Fenchel-Young inequality: 

\begin{equation}\label{eq:fyineq}
\psi(\theta(p))+\phi(\eta(q))\geq \sum_{i=1}^m \theta_i(p)\,\eta_i(q),
\end{equation}
 with equality if and only if $p=q$ on $M$.
The metric tensor $g$ can be expressed as $g=\nabla d\psi$ or equivalently as $g=\nabla^*d\phi$ where $d$ denotes the exterior derivative and $\psi$ and $\phi$ are $0$-forms.

Now, the $\nabla$-affine coordinate system $\theta$ is defined up  to an affine transformation: 
That is, if $\theta(\cdot)$ is a $\nabla$-coordinate system so is the coordinate system $\bar\theta(\cdot)=A\theta(\cdot)+b$. 
However, once parameters $A$ and $b$ are fixed, it fully determines the dual $\nabla^*$-coordinate system $\eta(\cdot)$.

The potential function $\phi$ is also reconstructed by solving differential equations modulo an affine term $\inner{c}{\theta}+d$ 
(e.g., see proofs relying on Poincar\'e lemma in~\cite{amari2000methods,Shima-2007,IG-2016,morales2023geometric}).
Fixing the degrees of freedom in the reconstruction of $\psi$ will also fix the corresponding affine term in $F^*$ which expresses $\psi$.
The dual potential functions $\psi$ and $\phi$ on the manifold $M$ are related by the fiberwise Legendre transform~\cite{Leok-2017} in geometric mechanics.

Thus the information geometry of dually flat spaces allows one to explain the decoupling of the interactions of the dual set of parameters $(A,b,c,d)$ with 
$(E,f,g,h)$.
Last, the scalar parameter $\lambda>0$ is the degree of freedom obtained from the fact that if $(M,g,\nabla,\nabla^*)$ is a dually flat space so is 
$(M,\lambda g,\nabla,\nabla^*)$. 

The Fenchel-Young inequality induced by the dual potential functions of Eq.~\ref{eq:fyineq} define a canonical divergence on a dually flat space that we term
 the dually flat  Hessian divergence:
$$
D_{g,\nabla,}(p:q)=\psi(\theta(p))+\phi(\eta(q)) - \sum_{i=1}^m \theta_i(p)\eta_i(q)\geq 0.
$$
We have the following reference duality~\cite{Naudts-2024}: $D_{g,\nabla}(q:p)=D_{g,\nabla^*}(p:q)$.

The dually flat  Hessian divergence can be expressed using the dual coordinate systems as a corresponding
Fenchel-Young divergence  defined by 
$$
Y_{F,F^*}(\theta:\eta') \eqdef F(\theta)+F^*(\eta')-\sum_{i=1}^m \theta_i\,\eta_i'.
$$

A Fenchel-Young divergence can also be expressed equivalently as dual Bregman divergences  $B_F$ or $B_{F^*}$:
$$
Y_{F,F^*}(\theta:\eta')=B_F(\theta:\theta')=B_{F^*}(\eta':\eta),
$$
 where $\theta'=\nabla F^*(\eta')$ and $\eta=\nabla F(\theta)$ with
$$
B_F(\theta:\theta')=F(\theta)-F(\theta')-\inner{\theta-\theta'}{\nabla F(\theta')}.
$$

Thus the dually flat Hessian divergence can be expressed as dual Fenchel-Young divergences (using the fact that ${F^*}^*=F$) in the mixed $\theta$/$\eta$ coordinate systems as
\begin{eqnarray}
D_{g,\nabla,}(p:q) &=& Y_{F,F^*}(\theta(p):\eta(q)) =  Y_{F^*,F}(\eta(q):\theta(p)),\\
&=&  \frac{1}{\lambda}\, Y_{F_P,F^*_{P^\diamond}}(\bar\theta(p):\bar\eta(q)) = \frac{1}{\lambda}\, Y_{F^*_{P^\diamond},F_P}(\bar\eta(q):\bar\theta(p)), \forall P\in \bbP,
\end{eqnarray}
This follows from the fact that $\psi(p)=F_P(\bar\theta(p))=F(\theta(p))$ and $\phi(q)=F^*_{P^\diamond}(\bar\eta(q))=F^*(\eta(q))$ for any $P\in\bbP$.

Thus we can define an equivalence relation $\sim$ between functions of $\Gamma_2$:
 $F\sim \tilde{F}$ if and only if $\tilde{F}=F_P$ for some $P\in\bbP$. 
Furthermore, a distance invariant under the Legendre transform~\cite{Attouch-1986} can be defined on the moduli space $\Gamma_2/\sim$ of dually flat spaces.

Last, we may consider some curvilinear dual coordinate systems instead of the mutually orthogonal dually affine coordinate systems: 
Let us transform the affine $\theta$-coordinate system into some arbitrary curvilinear coordinate system $\tilde\theta$.
The  dual affine $\eta$-coordinate system transforms correspondingly into the curvilinear coordinate system $\tilde\eta$.
The underlying geometric structures have been called the $(u,v)$-structures by Amari~\cite{uv-2015,IG-2016} or the $(\rho,\tau)$-structures by Zhang~\cite{Naudts-2024}.
The   dual Bregman divergences or equivalently the dual Fenchel-Young divergences of the underlying dually flat space need to be called by undeforming arguments~\cite{nielsen2009dual,amari2009alpha} $\tilde\theta$ and/or $\tilde\eta$. 
Since the inverse transform $\tilde\theta\rightarrow\theta$ can be interpreted as a representation function, the Bregman divergences called by undeforming $\tilde\theta$ arguments have been called representational Bregman divergences in~\cite{nielsen2009dual}.
.

\section{Conclusion}

To summarize, we first proved that the generalized Legendre transforms obtained from the reverse-ordering involutive transform axiomatization of~\cite{GenLegendreFenchel-2007,AxiomatizationLegendreTransform-2009}  can be interpreted as the ordinary Legendre transform on corresponding affine-deformed functions.
Second, we explained how these generalized Legendre transforms 
are merely different expressions of the well-known  geometric Legendre transforms on  dually flat spaces (Hessian manifolds with global charts).

Figure~\ref{fig:LegendreDivIG} summarizes the various classes of convex functions used in this paper with some of their key properties used to define information-geometric structures. 
In particular, the class of functions $\Gamma_3$ are thrice-differentiable Legendre-type lsc. convex functions (with $\Gamma_3\subset\Gamma_2\subset\Gamma_1\subset\Gamma_0$) which allows one to construct $\alpha$-geometry~\cite{IG-2016} from the Amari-Chentsov cubic tensor~\cite{ay2017information}.

\begin{figure}
\centering
\includegraphics[width=\textwidth]{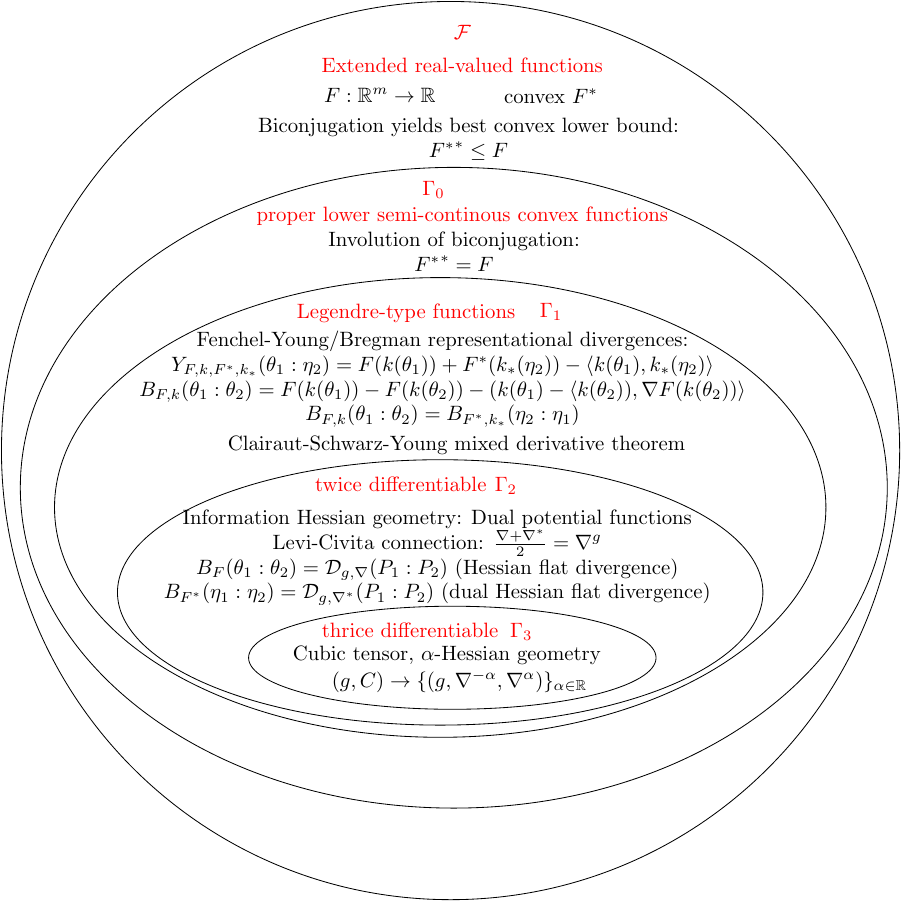}

\caption{The ordinary Legendre transform on various classes of functions: Relationships with  Fenchel-Young and Bregman divergences, dually flat Hessian divergence, and $\alpha$-geometry in information geometry.\label{fig:LegendreDivIG}}
\end{figure}


\appendix

\section{Legendre-type functions}\label{sec:ltf}

 Rockafellar~\cite{LegendreType-1967} showed that there exist convex lsc. functions of $\Gamma_0$  with corresponding gradient domains not convex.
For example, Rockafellar showed that the bivariate function  $F(\theta)=F(\theta_1,\theta_2) = \frac{1}{4} \left( \frac{\theta_1^2}{\theta_2} + \theta_1^2 +\theta_2^2\right)$
defined on the upper plane domain $\Theta=\bbR\times\bbR_{>0}$ is strictly convex but its gradient domain $H$ is not convex.
Thus to rule out these functions, Rockafellar defined Legendre-type functions~\cite{LegendreType-1967}:

\begin{Definition}[Legendre-type function~\cite{LegendreType-1967}]
A convex function $F:\Theta\subset\bbR^m\rightarrow\barR$ is of Legendre-type if

\begin{itemize}
\item $\Theta$ is a non-empty open effective domain,

\item $F$ is strictly convex and differentiable on $\Theta$,

\item $F$ becomes infinitely steep close to boundary points of its effective domain:

$$
\forall \theta\in\Theta, \forall \theta'\in\partial\Theta,\quad \lim_{\lambda\rightarrow 0}\,
 \frac{d}{d\lambda} F(\lambda\theta+(1-\lambda)\theta') = -\infty.
$$
\end{itemize}
\end{Definition}

When $F$ is of Legendre-type so is its convex conjugate $F^*$, and moreover the gradients of convex conjugates are reciprocal to each others: 
$\nabla F^*=(\nabla F)^{-1}$ and $\nabla F=(\nabla F^*)^{-1}$. 
This property of reciprocal gradients obtained from the LFT of Legendre-type functions is strong
 since in general the implicit function theorem  only guarantees local inversion of   multivariate functions but not the existence of global inverse functions.

We shall denote by $\Gamma_1(\Theta)\subset \Gamma_0(\Theta)$ the subset of proper lsc. strictly convex and differentiable functions of Legendre-type defined on the open effective domain $\Theta$.
We  define the Legendre-Fenchel transform of $\calL (\Theta,F)=(H,F^*)$ where $H$ is the gradient domain, and called $F^*$ the convex conjugate of $F$.
We have $\calL\calL (\Theta,F)=(\Theta,F)$ for Legendre-type functions $F\in\Gamma_1(\Theta)$.
The concept of Legendre-type function is related to the concept of steep exponential families in statistics.

In general, convex conjugate pairs enjoy the following reverse-ordering property:

\begin{Property}[Reverse-ordering]\label{prop:revorder}
Let $F_1$ and $F_2$ in $\Gamma_0(\Theta)$.
If $F_2\leq F_1$ then $\calL F_2\geq \calL F_1$.
If $F_2\geq F_1$ then $\calL F_2\leq \calL F_1$.
\end{Property}

\begin{proof}
Assume $F_2\leq F_1$ and let us prove that $F_2^*=\calL F_2\geq F_1^*=\calL F_1$:

\begin{eqnarray*}
F_1^*(\eta) &=& \sup_\theta \{\inner{\eta}{\theta}-F_1(\theta)\},\\
  &=& \inner{\eta}{\nabla F^{-1}(\eta)} - F_1(\nabla F^{-1}(\eta)),\\
	&\leq & \inner{\eta}{\nabla F^{-1}(\eta)} - F_2(\nabla F^{-1}(\eta)),\\
	&\leq & \sup_\theta \{\inner{\eta}{\theta} - F_2(\theta)\} = F_2^*(\eta).
\end{eqnarray*}

The case $F_2\geq F_1$ is similar to the case $F_1\geq F_2$ by exchanging the role of $F_1$ and $F_2$ (i.e., $F_2\leftrightarrow F_1$).
\end{proof}

\section{Some examples of convex conjugates}\label{sec:examples}

\begin{Example}
Let $F(\theta)=\inner{a}{\theta}+b$ be an affine function.
Then 
$$
F^*(\eta)=\left\{\begin{array}{ll}-b & \mbox{if $\eta=a$}\cr +\infty & \mbox{if $\eta\not=a$}\end{array}\right.
$$
Let 
$$
1_A(x)=\left\{
\begin{array}{ll} 
0 & \mbox{if $x\in A$}\cr 
+\infty & \mbox{if $x\not\in A$}
\end{array}\right.
$$ 
be the indicator function of a set $A$.
Then we can write $F^*=1_{\{a\}}-b$.
\end{Example}

\begin{Example}
Let $F(\theta)=|\theta|$ for $\theta\in\bbR$.
Then we have the convex conjugate 
$$
F^*(\eta)=1_{[-1,1]}(\eta)=\left\{
\begin{array}{ll}
0 & |\eta|\leq 1,\\
+\infty & \mbox{otherwise}.
\end{array}
\right.
$$
\end{Example}

\begin{Example}
Consider $F(\theta)=\exp\theta$ for $\Theta=\bbR$.
We have
$$
F^*(\eta)=\left\{
\begin{array}{ll}
\eta\log\eta-\eta, & \mbox{if $\eta>0$},\\
0,&\mbox{if $\eta=0$},\\
+\infty & \mbox{if $\eta<0$}.
\end{array}
\right.
$$
The convex conjugate is the scalar Shannon entropy function extended to positive reals.
\end{Example}

\begin{Example}
For $p\in [1,+\infty)$, let $F_p(\theta)=\frac{1}{p} \|\theta\|^p$ for $\theta\in\Theta=\bbR^m$.
Then we have $F_p^*(\eta):=F_q(\eta)$ where $\frac{1}{p}+\frac{1}{q}=1$.
The powers $p$ and $q$ are called a conjugate pair of exponents.
Notice that the Fenchel-Young inequality reduces to the Young's inequality in the scalar case:
$$
F_p(\theta')+F_q(\eta)\geq \inner{\theta'}{\eta}.
$$
Furthermore, by integrating Young's inequality on both sides for $\eta=g(x)$ and $\theta=f(x)$ for $x\in\calX$ and $f\in L_p$ and $g\in L_q$ (Lebesgue spaces), we recover H\"older's inequality:
$$
\frac{1}{p}\|f\|_p+\frac{1}{q}\|g\|_q \geq \|fg\|_1.
$$
In fact, we have the self-duality $F=F^*$ only for $F(\theta)=F_2(\theta)=\frac{1}{2} \|\theta\|^2$. 
\end{Example}

When a function $F\in\Gamma_0(\Theta)$ is not Legendre type, the Fenchel-Young inequality may be tight for several pairs $(\theta',\eta)$ instead of a single pair 
$(\theta=\nabla F^*(\eta),\eta=\nabla F(\theta))$.
A subgradient $\eta$ of $F(\theta)$ at $\theta_0$ satisfies:
$$
F(\theta)-F(\theta_0)\geq \inner{\eta}{\theta-\theta_0}.
$$
The subdifferential $\partial_{\theta_0} F$ of $F$ at $\theta_0$ is the set of subgradients.
The subdifferential operator $\partial:\bbR^m\rightrightarrows \bbR^m$ is multivalued:
$$
\partial F(\theta)=\left\{
\begin{array}{ll}
\left\{\eta \st F(\theta)-F(\theta_0)\geq \inner{\eta}{\theta-\theta_0} \right\}, & \mbox{if $\theta\in\Theta=\dom(F)$}\\
\emptyset & \mbox{if $\theta\not\in\Theta$}.
\end{array}
\right.
$$ 

The following example illustrates the LFT of a non-Legendre type function and the LFT of a Legendre-type function obtained by restricting the function to a specified domain.

\begin{figure}
\centering
\includegraphics[width=0.8\textwidth]{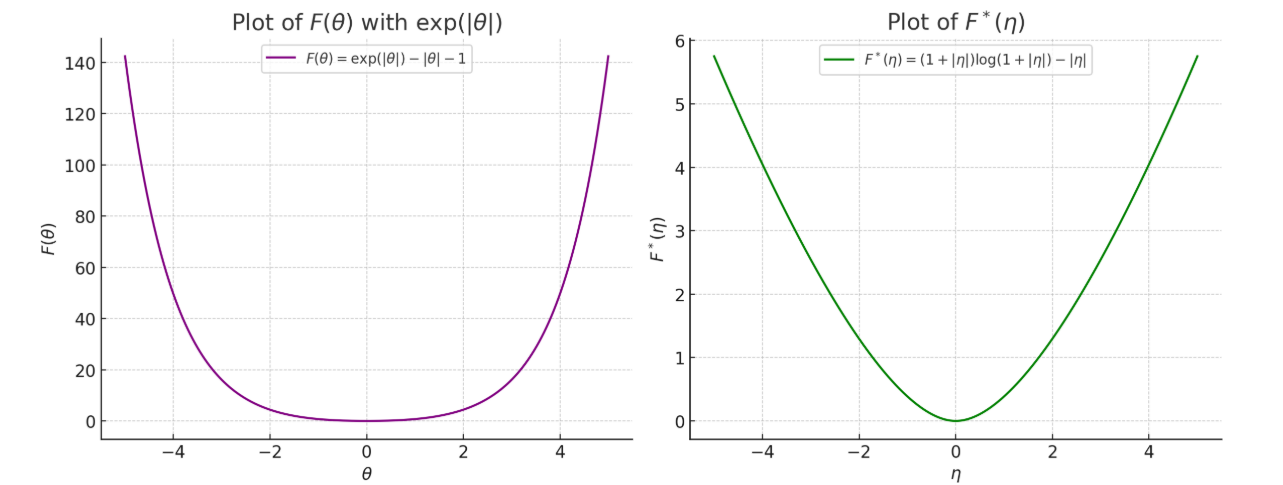}\\
\includegraphics[width=0.8\textwidth]{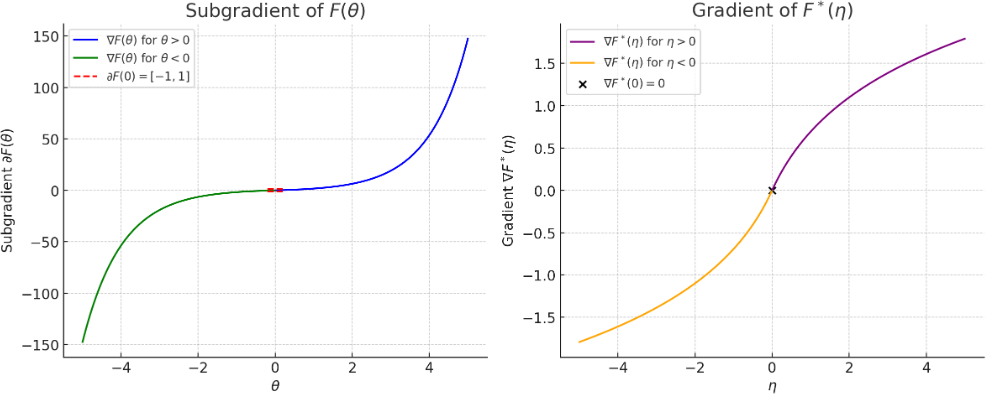}
\caption{A pair $(F(\theta),F^*(\eta))$ of conjugate functions (top) with their subgradients plotted (bottom). 
Function $F(\theta)$ is not differentiable at $\theta=0$ and thus admits a subgradient $\partial F(0)$ at $\theta=0$. 
Function $F^*(\eta)$ is everywhere differentiable, and when $\theta\not=0$, we have $\nabla F^*=(\nabla F)^{-1}$ visualized by rotating the $xy$-axis 90-degrees).
\label{fig:LTex1}}
\end{figure}

\begin{Example}
Let $F(\theta)\eqdef\exp(|\theta|)-|\theta|-1$ defined on $\Theta=\bbR$. Then $F^*(\eta)=(1+|\eta|)\log(1+|\eta|)-|\eta|$ defined on $H=\bbR$.
These two functions $F$ and $F^*$ are even and convex, and form a conjugate pair (Figure~\ref{fig:LTex1}).
Function $F(\theta)$ is not differentiable at $\theta=0$, and thus we have:
$$
\partial F(\theta)=\left\{
\begin{array}{ll}
\nabla F(\theta)=\exp(\theta)-1, & \theta>0,\\
\partial F(0)=[-1,1], & \theta= 0,\\
\nabla F(\theta)=-\exp(-\theta)+1, & \theta<0.
\end{array}
\right.
$$
Function $F^*(\eta)$ is differentiable everywhere and we have
$$
\partial F^*(\eta)=\left\{
\begin{array}{ll}
\nabla F^*(\eta)=\log(1+\eta), & \eta\geq 0,\\
\nabla F^*(\eta)=-\log(1-\eta), & \eta\leq 0.
\end{array}
\right.
$$

For $\theta\not=0$, the gradients $\nabla F(\theta)$ and $\nabla F^*(\eta)$ are reciprocal to each other.
The function $(\bbR,F(\theta))$ is not Legendre-type since it  is not differentiable at $\theta=0$ but the functions $(\bbR_{>0},F(\theta))$ and $(\bbR_{>0},F^*(\eta))$ form a Legendre-type pair when restricted to the positive domain $\bbR_{>0}$: $F,F^*\in \Gamma_1(\bbR_{>0})$.
\end{Example}

\end{document}